\documentclass{article}
\usepackage[latin1]{inputenc}
\usepackage{amssymb, amsmath, amsthm}
\usepackage[a4paper]{geometry}
\usepackage{authblk} 
\usepackage{pifont}
\usepackage{graphicx}
\usepackage{xcolor}
\usepackage{xtab} 
\usepackage{longtable} 
\usepackage{footnote}
\makesavenoteenv{tabular}
\usepackage{tabularx}
\usepackage{booktabs}
\newcolumntype{J}{>{\raggedleft\arraybackslash}X}
\newcolumntype{B}{>{\raggedright\arraybackslash}X}
\usepackage{fullpage} 
\usepackage{dcolumn} 
\newcolumntype{.}{D{.}{.}{-1}}

\usepackage{tikz}
\usetikzlibrary{arrows}
\usepackage[round]{natbib}
\usepackage{hyperref}
\hypersetup{
  colorlinks   = true,    
  urlcolor     = blue,    
  linkcolor    = blue,    
  citecolor    = red      
}

\usepackage{enumitem}

\setlist[itemize]{leftmargin=3\parindent}
\setlist[enumerate]{leftmargin=2\parindent}

\theoremstyle{plain}

\newtheorem{corollary}{Corollary}
\newtheorem{definition}{Definition}[section]
\newtheorem{example}{Example}[section]
\newtheorem{lemma}{Lemma}[section]

\newtheorem{proposition}{Proposition}[section]
\newtheorem{remark}{Remark}
\newtheorem{theorem}{Theorem}[section]

\makeatletter
\renewcommand*\env@matrix[1][c]{\hskip -\arraycolsep
  \let\@ifnextchar\new@ifnextchar
  \array{*\c@MaxMatrixCols #1}}
\makeatother

\def\keywords{\vspace{.5em} 
{\textit{Keywords}:\,\relax%
}}

\author{L\'aszl\'o Csat\'o\thanks{e-mail: laszlo.csato@uni-corvinus.hu} }
\affil{Department of Operations Research and Actuarial Sciences \\ Corvinus University of Budapest \\ MTA-BCE ''Lend\"ulet'' Strategic Interactions Research Group \\ Budapest, Hungary}
\title{A graph interpretation of the least squares ranking method}
\date{\today}

\begin{document}

\maketitle

\begin{abstract}
The paper aims at analyzing the least squares ranking method for generalized tournaments with possible missing  and multiple paired comparisons. The bilateral relationships may reflect the outcomes of a sport competition, product comparisons, or evaluation of political candidates and policies. It is shown that the rating vector can be obtained as a limit point of an iterative process based on the scores in almost all cases. The calculation is interpreted on an undirected graph with loops attached to some nodes, revealing that the procedure takes into account not only the given object's results but also the strength of objects compared with it. We explore the connection between this method and another procedure defined for ranking the nodes in a digraph, the positional power measure. The decomposition of the least squares solution offers a number of ways to modify the method.

\end{abstract}

\section{Introduction}

Ranking of alternatives is becoming an important tool for individuals, enterprises and nonprofit organizations to help their decision making processes. 
In various contexts the necessary information is available as outcomes of paired comparisons of the objects. Problems of this kind appear in social choice theory, statistics \citep{EltetoKoves1964, Szulc1964}, sports \citep{Landau1895, Landau1914, Zermelo1928} or psychology \citep{Thurstone1927}, to name a few.

There exist two fundamentally different approaches in ranking methodology. The first one seeks various scoring functions, giving a weight or valuation to all alternatives, that is, they compress the paired comparison matrix into a single rating vector. The second approach is based on the approximation of the (generalized) tournament by linear orders \citep{Kemeny1959, Slater1961}, which usually leads to interesting combinatorial and algorithmic problems \citep{Hudry2009}. From a theoretical viewpoint these methods have two great disadvantages: the possible occurrence of multiple optimal solutions and the difficulties arising in the examination of their (normative) properties \citep{Bouyssou2004}. Consequently, we will follow the former approach.

Score seems to be the most obvious rating method: it is obtained by adding the number of victories for each object. It is an appropriate choice in the case of complete tournaments such that all objects have set against each other at the same number of occasions. However, there are many situations where it is unfeasible to get direct information about each pair of alternatives. It implies that the schedule becomes important since an object compared with weak opponents may score more victories than its peers facing stronger objects. In this case, for example for Swiss system tournaments, the application of scores in order to rank the objects can be questioned \citep{Csato2013a}.

In order to take into account the quality of opponents, a large number of scoring procedures have been suggested, see, for instance, \citet{ChebotarevShamis1998a} for a survey of them.
\citet{ChebotarevShamis1999} introduces two classes, the \emph{win-loss combining} and \emph{win-loss unifying} procedures, to categorize the methods proposed in the literature.
Win-loss combining procedures can usually be calculated iteratively on a graph, where the vertices represent the alternatives and the edges reflect the paired comparisons. Among them, PageRank is one of the most popular \citep{BrinPage1998}. \citet{SlikkerBormvandenBrink2012} integrated its core, the invariant method \citep{Daniels1969, MoonPullman1970} as well as fair bets \citep{Daniels1969, MoonPullman1970} and $\lambda$ \citep{BormvandenBrinkSlikker2002} methods into a single framework, and interpreted them on graphs. \citet{HeringsvanderLaanTalman2005} defines another method, the positional power function for ranking the nodes of directed graphs.

For the other class of win-loss unifying procedures, graph interpretation becomes more difficult since they treat wins and losses uniformly, therefore the direction of edges do not count. We know only one result in this field, the iterative calculation of a subclass of the generalized row sum \citep{Chebotarev1989, Chebotarev1994}, a parametric family of ranking methods \citep{Shamis1994}. This paper gives a graph interpretation for the least squares method \citep{Horst1932, Morrissey1955, Mosteller1951, Gulliksen1956, KaiserSerlin1978} through the use of scores and the comparison structure.

The iterative calculation is similar to the positional power \citep{HeringsvanderLaanTalman2005}, which is somewhat surprising, since the least squares method is defined as an optimization problem, not as an intuition-based proposal.
However, there are two main differences besides the applied approach (win-loss unifying vs win-loss combining):
the role of initial scores and iterated ratings in the calculation, and the choice of the parameter reflecting the importance of successors in digraphs or objects compared with the given one.

The graph interpretation also provides a lot of possibilities to modify it, in order to eliminate its drawbacks from an axiomatic viewpoint \citep{Gonzalez-DiazHendrickxLohmann2013}. We hope that simple calculation and evident connection with the scores can inspire practical applications, as well as offer an alternative to the extensively used PageRank method in certain cases. Nevertheless, it should be taken into account that scientometrics differs from our tournament context since a citation from a journal to another is advantageous for the latter, but not necessarily unfavourable for the former.


The paper is organized as follows. The setting is presented in Section \ref{Sec2}, where two scoring procedures are introduced, too. Section \ref{Sec3} deals with the least squares method, for which a new iterative solution method is given.
It is used for the decomposition of the rating vector leading to the graph interpretation, discussed in Section \ref{Sec4}. Finally, Section \ref{Sec5} concludes our results.

\section{Notations and rating methods} \label{Sec2}

Let $N = \{ X_1,X_2, \dots, X_n \}$, $n \in \mathbb{N}$ be a set of objects  and $\mathbf{R} = \left( R^{(1)}, R^{(2)}, \dots , R^{(m)} \right)$, $m \in \mathbb{N}$ be an array representing the outcomes of paired comparisons between the objects, where $R^{(p)}$ ($p = 1,2, \dots , m$) is an $n \times n$ nonnegative matrix corresponding to the $p$th experiment, round of tournament, questionnary etc. Matrices $R^{(p)}$ may be defined partially, $r_{ij}^{(p)}$ and $r_{ji}^{(p)}$ remain unknown if objects $X_i$ and $X_j$ were not compared in the $p$th round. A fully defined matrix $R^{(p)}$ is called a \emph{complete} paired comparison matrix, while the one with some 'missing' elements is \emph{incomplete}. For all pairs of objects compared $(X_i,X_j)$, $r_{ij}^{(p)} + r_{ji}^{(p)} = 1$ is assumed. 
$r_{ij}^{(p)}$ can be interpreted as the likelihood assigned to the event $X_i$ is better than $X_j$ in the $p$th round of a tournament. Diagonal elements $r_{ii}$ are supposed to be $0$ for all $i = 1,2, \dots ,n$, but they will not be used for the ranking methods discussed.

Most scoring methods are based on the \emph{aggregated paired comparison matrix} $R = (r_{ij})$ containing the sum of the results for all pairs of objects:
\[
r_{ij} = \left\{
\begin{array}{ll}
0 & \textrm{if } r_{ij}^{(p)} \textrm{ is not defined for every } p = 1,2, \dots ,m \\
\sum_{p=1, \, r_{ij}^{(p)} \text{ is defined}}^m r_{ij}^{(p)} & \textrm{otherwise}.
\end{array} \right.
\]
Generally, the outcomes $r_{ij}^{(p)}$ can be aggregated by taking a weighted sum, thus associating different weights on various rounds/experts/areas etc. It makes sense in forecasting sport results, when the latest paired comparisons are considered to be more important.

The pair $(N,\mathbf{R})$ is called a \emph{preference profile}. The set of preference profiles is $\mathcal{R}$. This setting is able to integrate four extra features in addition to those of binary tournaments (complete, weak and asymmetric binary relations, see \citep{Rubinstein1980}):
\begin{itemize}[label=$\bullet$]
\item
the possibility of ties: $r_{ij} = r_{ji}$;
\item
different preference intensities captured by the likelihood $r_{ij} / (r_{ij} + r_{ji})$;
\item
incompleteness as $r_{ij}$ is undefined (unknown) for some pairs of objects $(X_i,X_j)$;
\item
multiple comparisons between objects: $r_{ij}^{(p)}$ is known for more than one $p$.
\end{itemize}
For example, if only strict binary relations are allowed, then $r_{ij}^{(p)} \in \{ 0,1 \}$ for all $p = 1,2, \dots ,m$.
This notation follows \citet{ChebotarevShamis1998a} and \citet{Gonzalez-DiazHendrickxLohmann2013}.

Regarding the four extensions of the original binary tournament model, the possibility of ties is an immediate consequence of different preference intensities, a common feature in many situations.
Multiple comparisons can arise naturally from the definition of preference profile. It can be supposed that incompleteness does not appear in an ideal case; however, we want to allow an expert to avoid judgement if he/she is not familiar with the two alternatives. The lack of direct information about paired comparisons may arise when the problem contains a large number of objects or it is too expensive to compare each pair; the latter is the reason of the emergence of knockout or Swiss system tournaments in some sports. Another case of incomplete comparisons is the need for predicting the final ranking before all rounds of a round-robin tournament played.

A \emph{rating (scoring) method} $f$ is an $\mathcal{R} \to \mathbb{R}^n$ function, $f_i = f(N,\mathbf{R})_i$ is the rating of object $X_i$.
It defines a corresponding \emph{ranking method} $\varphi$, that is, a (transitive and complete) weak order such that the objects are arranged according to $f$: $\varphi$ ranks $X_i$ weakly above $X_j$ if and only if $f_i \geq f_j$. It is denoted by $X_i \succeq_{\varphi} X_j$ or simply by $X_i \succeq X_j$, if it is not misleading. This definition of $\succeq$ already determines that $X_i$ is ranked strictly above $X_j$ if and only if $X_i$ is ranked weakly above $X_j$, but $X_j$ is not ranked weakly above $X_i$: $(X_i \succ X_j) \Leftrightarrow \left[ (X_i \succeq X_j) \text{ and } \neg (X_j \succeq X_i) \right]$. Finally, the ranking can be tied between objects $X_i$ and $X_j$: $X_i \sim X_j \Leftrightarrow \left[ (X_i \succeq X_j) \text{ and } (X_j \succeq X_i) \right]$. Ratings give cardinal, and rankings give ordinal information about the objects. Throughout the paper, the notions of rating and ranking methods will be used analogously since the discussed ranking procedures are based on rating vectors.

A scoring procedure is \emph{neutral} if any reindexing of the objects in $N$ preserves their rating.
A scoring procedure is \emph{anonymous} if any reindexing of the paired comparison matrices $R^{(p)}$ in an array $\mathbf{R}$ preserves the ratings of objects.
It seems to be quite natural to demand neutrality and anonymity; all rating procedures discussed here will satisfy these conditions. Note that a method based on the aggregated paired comparison matrix $R$ is always anonymous. Rating procedures $f_1$ and $f_2$ are called \emph{equivalent} if they result in the same ranking.

Ranking of the objects involves two main challenges. The first one is common in all paired comparison models: the possible appearance of \emph{circular triads}, when object $X_i$ is better than $X_j$ (that is, $r_{ij} > r_{ji}$), $X_j$ is better than $X_k$, but $X_k$ is better than $X_i$. Circular triads generate difficulties in all paired comparison settings, but, if preference intensities also count, other triplets ($X_i, X_j, X_k$) may produce problems. The second issue arises as the consequence of incomplete and multiple comparisons: the performance of objects compared with $X_i$ strongly influences the observable paired comparison outcomes $r_{ij}$. For example, if $X_i$ was compared only with $X_j$, then its rating certainly should depend on the results of $X_j$. We will see that this argument can be continued infinitely. Since both problems can occur only if there is at least three objects, the case $n=2$ becomes trivial.

An alternative representation of paired comparisons is the following.
The \emph{additive paired comparison matrix} $A^{(p)}$ can be derived from $R^{(p)}$ by 'centering' the outcomes of paired comparisons such that $a_{ij}^{(p)} = r_{ij}^{(p)} - r_{ji}^{(p)}$. For undefined comparisons $a_{ij}^{(p)}$ is set to $0$. Now $a_{ji}^{(p)} = -a_{ij}^{(p)}$, $A^{(p)}$ is a skew-symmetric matrix. It is called \emph{consistent} if $a_{ij} = a_{ik} + a_{kj}$ for all triplets $(X_i,X_j,X_k)$, and \emph{inconsistent} if this condition is not satisfied for some $(X_i,X_j,X_k)$. The aggregated additive paired comparison matrix $A = (a_{ij})$ is defined analogously by $A = \sum_{p=1}^m A^{(p)}$, and will be referred to as the \emph{results matrix}.

The numbers of comparisons between the objects determine the \emph{matches matrix} $M = (m_{ij})$:
\[
m_{ij} = \left\{
\begin{array}{ll}
\textrm{the number of indices $1 \leq p \leq m$ such that $r_{ij}^{(p)}$ is defined} & \textrm{if } i \neq j \\
0 & \textrm{if } i=j.
\end{array} \right.
\]
$M$ is a symmetric matrix and $0 \leq m_{ij} \leq m$. It is not restrictive to assume that $m = \max_{i,j} m_{ij}$ if the reduced matrix $A$ is analyzed. In most practical applications (and in our setting above) $m_{ij} \in \mathbb{N}$, but the whole discussion is valid for $m_{ij} \in \mathbb{R}_+$ as well, this domain choice has no impact on the results. The generalization has some significance for example in the above mentioned problem of forecasting sport results. Here the latest comparisons contain more information about the current form of the player, which may be addressed by exponential smoothing, a technique usually applied to time series data.

The matches matrix $M$ is called \emph{block diagonal} and \emph{block anti-diagonal}, respectively, if it has a partition $N_1 \cup N_2 = N$, $|N_1| = n_1$ and $|N_2| = n_2$ such that with a possible reordering of the objects
\[
M = \left(
\begin{array}{cc}
M^1_{n_1 \times n_1} & \mathbf{0}_{n_1 \times n_2} \\
\mathbf{0}_{n_2 \times n_1} & M^2_{n_2 \times n_2}
\end{array} \right)
\textrm{ and }
M = \left(
\begin{array}{cc}
\mathbf{0}_{n_1 \times n_1} & M^1_{n_1 \times n_2} \\
M^2_{n_2 \times n_1} & \mathbf{0}_{n_2 \times n_2}
\end{array} \right),
\]
respectively, where the subscripts denote the dimensions of (sub)matrices. Furthermore, $d_i = \sum_{j=1}^n m_{ij}$ is the \emph{total number of comparisons} of object $X_i$.

Results matrix $A$ and matches matrix $M$ together with the set of objects $N$ determine a \emph{ranking problem} $(N,A,M)$ or $(A,M)$ for short.
In this modified setting, $(a_{ij} + m_{ij})/(2m_{ij}) \in \left[ 0,1 \right]$ may be regarded as the likelihood that object $X_i$ defeats $X_j$.

A ranking problem is called \emph{round-robin} if $m_{ij} = 1$ for all $i \neq j$, that is, every object has been compared with all the others exactly once and $d_i = n-1$ for all $i = 1,2, \dots ,n$. A round-robin ranking problem is more general than the binary tournaments of \citet{Rubinstein1980} as it allows for ties ($a_{ij} = a_{ji} = 0$) and preference intensities ($a_{ij}$ is not necessarily $-1$ or $1$).
A ranking problem is called \emph{unweighted} if $m_{ij} \in \{ 0,1 \}$ for all $i \neq j$, namely, every paired comparison is carried out at most once. Otherwise the ranking problem is called \emph{weighted}.

Matrix $M$ can be represented by an undirected multigraph $G := (V,E)$ where vertex set $V$ corresponds to the object set $N$, and the number of edges between objects $X_i$ and $X_j$ is equal to $m_{ij}$. Therefore the set of edges represents the structure of known paired comparisons. The number of edges adjacent to $X_i \in N$ is the \emph{degree} $d_i$ of node $X_i$. A \emph{path} is a sequence of objects $X_{k_1}, X_{k_2}, \dots , X_{k_t}$ such that $m_{k_\ell k_{\ell+1}} > 0$ for all $\ell = 1,2, \dots ,t-1$. Two vertices are \emph{connected} if $G$ contains a path between them. A graph is said to be connected if every pair of vertices is connected. The \emph{adjacency matrix} $T^A$ of $G$ is given with the elements $t_{ij} = 1$ if $m_{ij} > 0$ and $t_{ij} = 0$ otherwise.

Graph $G$ is called the \emph{comparison multigraph} associated with the ranking problem $(N,A,M)$, however, it is independent of the results of paired comparisons. The \emph{Laplacian matrix} $L = \left[ \ell_{ij} \right], \, i,j = 1,2, \dots ,n$ of graph $G$ is an $n \times n$ real matrix with $\ell_{ij} = -m_{ij}$ for all $i \neq j$ and $\ell_{ii} = d_i$ for all $i = 1,2, \dots ,n$. $L$ has real and nonnegative eigenvalues (it is positive semidefinite) \citep[Theorem 2.1]{Mohar1991}, denoted by $\mu_1 \geq \mu_2 \geq \dots \geq \mu_{n-1} \geq \mu_{n} = 0$. Let $\mathbf{e} \in \mathbb{R}^n$ be the unit column vector, that is, $e_i = 1$ for all $i = 1,2, \dots ,n$.

Now we define two rating methods for the ranking problem $(N,A,M)$.

\begin{definition}
\textbf{Row sum} rating method: $\mathbf{s} = \sum_{p=1}^m \mathbf{s}^{(p)} = \sum_{p=1}^m A^{(p)} \mathbf{e} = A \mathbf{e}$.
\end{definition}

Row sum will also be referred to as \emph{scores}, $\mathbf{s}$ is sometimes called the \emph{scores vector}.
The following parametric rating procedure was constructed axiomatically by \citet{Chebotarev1989} and thoroughly analyzed in \citet{Chebotarev1994}.

\begin{definition} \label{Def23}
\textbf{Generalized row sum} rating method: it is the unique solution $\mathbf{x}(\varepsilon)$ of the system of linear equations $(I+ \varepsilon L) \mathbf{x}(\varepsilon) = (1 + \varepsilon m n)\mathbf{s}$, where $\varepsilon \geq 0$ is a parameter, $\mathbf{s}$ is the scores vector, $I$ is the $n \times n$ identity matrix, and $L$ is the Laplacian matrix of the comparison multigraph $G$.
\end{definition}

It follows from the definition that this procedure results in the row sum ranking if $\varepsilon = 0$. For larger parameter values it adjusts the standard scores of objects by accounting for the performance of objects compared with it, and so on. $\varepsilon$ indicates the importance attributed to this correction of scores $\mathbf{s}$.

Both the score and the generalized row sum ratings are well-defined and easily computable from a system of linear equations for all ranking problems $(A,M)$.

\section{The least squares method and its solution} \label{Sec3}

Another approach to ranking is the statistical estimation by identifying $h_{ij} = 2a_{ij}/m_{ij}$ as the realized difference between the latent valuations of objects $X_i$ and $X_j$. In the ideal case no randomness is present and there exists a rating vector $q \in \mathbb{R}^n$ such that $h_{ij} = q_i - q_j$ for all pairs of objects $(X_i,X_j)$. It requires the consistency of the results matrix $A$ since $0 = (q_i - q_j) + (q_j - q_k) + (q_k - q_i) = h_{ij} + h_{jk} + h_{ki}$ for all $(X_i,X_j,X_k)$. If it is inconsistent, the actual outcome $h_{ij}$ may differ from its 'expected value' $q_i - q_j$, and it makes sense to apply the following least squared error minimization
\[
\min_{\mathbf{q} \in \mathbb{R}^n} \sum_{X_i,X_j \in N} m_{ij} (h_{ij} - q_i + q_j)^2.
\]

This method was discussed by \citet{Horst1932} and \citet{Mosteller1951} in the round-robin case, and was extended to unweighted problems by \citet{Gulliksen1956} and \citet{KaiserSerlin1978}. The weighted case is examined in \citet{BozokiCsatoRonyaiTapolcai2014} and \citet{Gonzalez-DiazHendrickxLohmann2013}, but it can also be regarded as unweighted by summation over indices $i,j,p$ \citep{ChebotarevShamis1999}. Clearly, the problem has an infinite number of solutions because the value of the objective function is the same for $\mathbf{q}$ and $\mathbf{q} + \beta \mathbf{e}, \beta \in \mathbb{R}$. A natural normalization is $\mathbf{e}^\top \mathbf{q} = 0$. The generalized row sum can be considered as a Bayesian modification of the least squares estimation \citep{Chebotarev1994}.

The first-order conditions of optimality give the following system of equations with unconstrained variables $q_i \in \mathbb{R}$ for all $i = 1,2, \dots ,n$:
\[
    \left(
    \begin{array}{cccccc}
    d_1       &   -m_{12}       &   -m_{13}  & \dots &   -m_{1,n-1} & -m_{1,n} \\
    -m_{12}      &   d_2       &   -m_{23} & \dots  &   -m_{2,n-1}  & -m_{2,n}    \\
    -m_{31} &   -m_{23} &   d_3       & \dots &  -m_{3,n-1}  & -m_{3,n} \\
    \vdots &   \vdots       & \vdots &  \ddots &   \vdots & \vdots      \\
    -m_{n-1,1} &   -m_{n-1,2}       &   -m_{n-1,3} & \dots &   d_{n-1} &  -m_{n-1,n}    \\
    -m_{n,1} &   -m_{n,2}       &  -m_{n,3} & \dots &   -m_{n,n-1} & d_{n}       \\
    \end{array}
    \right)
    \left(
    \begin{array}{c}
    q_1 \\
    q_2 \\
    q_3 \\
    \vdots \\
    q_{n-1} \\
    q_{n} \\
    \end{array}
    \right)
   =
    \left(
    \begin{array}{c}
    s_1 \\
    s_2 \\
    s_3 \\
    \vdots \\
    s_{n-1} \\
    s_{n} \\
    \end{array}
    \right),
\]
where $d_i =  \sum_{j=1}^n m_{ij}$ denotes the total number of $X_i$'s comparisons, and the element in the $(i,j)$ position ($i \neq j$) of the coefficient matrix equals $-m_{ij}$. On the right-hand side, $s_i = \sum_{j=1}^n a_{ij}$ is the score of object $X_i$. Due to the convexity of the objective function, the system of linear equations is a sufficient condition for optimality.

Note that the $n \times n$ matrix on the left-hand side is exactly the Laplacian matrix associated with the comparison multigraph, thus the first-order conditions give $L \mathbf{q} = \mathbf{s}$. $L$ has no inverse as sum of its rows (and columns) is zero.

\begin{definition}
\textbf{Least squares} rating method: it is the solution $\mathbf{q}$ of the system of linear equations $L \mathbf{q} = \mathbf{s}$ and $\mathbf{e}^\top \mathbf{q} = 0$.
\end{definition}

\begin{corollary} \label{Col1}
The least squares rating can be obtained as a limit of the generalized row sum method if $\varepsilon \to \infty$.
\end{corollary}

\begin{proof}
See \citet[p. 326]{ChebotarevShamis1998a}. 
\end{proof}

\begin{figure}
\centering
\caption{The preference graph of Example \ref{Examp31}}
\label{Fig1}
\vspace{0.5cm}
\begin{tikzpicture}[scale=1,auto=center, transform shape, >=triangle 45]
\tikzstyle{every node}=[draw,shape=rectangle];
  \node (n1) at (1,10) {$X_1$};
  \node (n3) at (1,8)  {$X_3$};
  \node (n5) at (1,6)  {$X_5$};
  \node (n2) at (7,10) {$X_2$};
  \node (n4) at (7,8)  {$X_4$};
  \node (n6) at (7,6)  {$X_6$};
  \node (n7) at (4,8)  {$X_7$};

  \foreach \from/\to in {n1/n3,n3/n5,n5/n6,n5/n7,n2/n4,n4/n6,n7/n6}
    \draw [->] (\from) -- (\to);
\end{tikzpicture}
\end{figure}
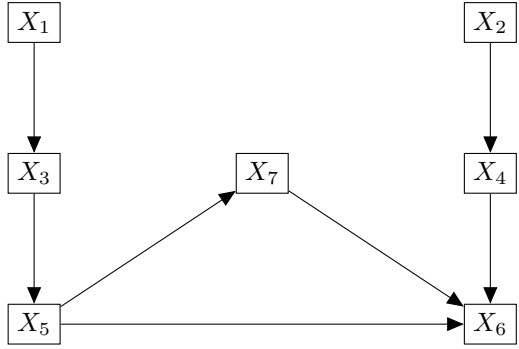

\begin{example} \label{Examp31}
\citep[Example 1]{Chebotarev1994} Suppose that the graph on Figure \ref{Fig1} is a \emph{preference graph}, reflecting the dominance relation between the objects: $a_{ij} = m_{ij} = 1$ if and only if there is an edge from $X_i$ to $X_j$. 

The corresponding results, matches matrices and the scores vector are as follows
\[
\setlength{\arraycolsep}{5pt}
A =
\begin{pmatrix}
    0     & 0     & 1     & 0     & 0     & 0     & 0 \\
    0     & 0     & 0     & 1     & 0     & 0     & 0 \\
    -1    & 0     & 0     & 0     & 1     & 0     & 0 \\
    0     & -1    & 0     & 0     & 0     & 1     & 0 \\
    0     & 0     & -1    & 0     & 0     & 1     & 1 \\
    0     & 0     & 0     & -1    & -1    & 0     & -1 \\
    0     & 0     & 0     & 0     & -1    & 1     & 0 \\
\end{pmatrix},
\quad
M =
\begin{pmatrix}
    0     & 0     & 1     & 0     & 0     & 0     & 0 \\
    0     & 0     & 0     & 1     & 0     & 0     & 0 \\
    1     & 0     & 0     & 0     & 1     & 0     & 0 \\
    0     & 1     & 0     & 0     & 0     & 1     & 0 \\
    0     & 0     & 1     & 0     & 0     & 1     & 1 \\
    0     & 0     & 0     & 1     & 1     & 0     & 1 \\
    0     & 0     & 0     & 0     & 1     & 1     & 0 \\
\end{pmatrix}
\text{ and }
\mathbf{s} =
\begin{pmatrix}[r]
    1 \\
    1 \\
    0 \\
    0 \\
    1 \\
    -3 \\
    0 \\
\end{pmatrix}.
\]
The solution for the least squares method is
\[
\setlength{\arraycolsep}{4pt}
\mathbf{q} =
\left[
\begin{array}{ccccccc}
    1.810 & 0.476 & 0.810 & -0.524 & -0.190 & -1.524 & -0.857 \\
\end{array}
\right]^\top .
\]
The scores method ($\mathbf{s}$) does not show the strength of objects compared with the given one. However, it is strange to assume that $X_1$ and $X_5$ have performed equally because the former has beaten the latter indirectly through $X_3$.
Least squares method results in the ordering $X_1 \succ X_3 \succ X_2 \succ X_5 \succ X_4 \succ X_7 \succ X_6$.
\end{example}

\citet[p. 127]{Gulliksen1956} notes for the unweighted case that, in general, the first minor of $L$ has an inverse, which makes it possible to normalize the rating vector by $q_n = 0$ and eliminate the last equation. After that, the upper-left $(n-1) \times (n-1)$ submatrix of $L$ denoted by $L_{-1}$ is taken with the corresponding first $n-1$ components of $\mathbf{s}$ and $\mathbf{q}$, denoted by $\mathbf{s}_{-1}$ and $\mathbf{q}_{-1}$, respectively. If $\left[ L_{-1} \right]^{-1}$ exists, then $\mathbf{q}_{-1} = \left[ L_{-1} \right]^{-1} \mathbf{s}_{-1}$. It means that the first $n-1$ equations of the system $L \mathbf{q} = \mathbf{s}$ remain to be satisfied if $q_n = 0$ is attached to $\mathbf{q}_{-1}$. The last equation is true because the sum of the first $n-1$ rows of $L$ is the opposite of the last, and similarly, the sum of the elements of $\mathbf{s}_{-1}$ is equal to $-s_n$.

In a round-robin ranking problem $\left[ L_{-1} \right]^{-1}$ can be computed explicitly, it is an $(n-1) \times (n-1)$ matrix with $2/n$ in each diagonal and $1/n$ in each off-diagonal entry \citep[p. 127]{Gulliksen1956}. The unique solution is $q_i = \sum_{j=1}^n a_{ij}/n = s_i/n$, implying that the row sum and least squares rankings coincide. This property is called \emph{score consistency} ($SCC$) by \citet{Gonzalez-DiazHendrickxLohmann2013}.

Some connections of the ranking problem and the associated Laplacian matrix $L$ are worth mentioning here.

\begin{lemma} \label{Lemma31}
For a ranking problem $(N,A,M)$, the following statements are equivalent:
	\begin{enumerate}
	\item \label{point1}
	Matches matrix $M$ is not block diagonal;
	\item \label{point2}
	Comparison multigraph $G$ is connected;
	\item \label{point3}
	The second smallest eigenvalue $\mu_{n-1}$ of $L$ is positive.	
	\end{enumerate}
\end{lemma}

\begin{proof}
The equivalence of items \ref{point1} and \ref{point2} is straightforward, since if $M$ is block diagonal, then there is no edge between the set of objects $N_1$ and $N_2$ and vice versa. \ref{point2} $\Leftrightarrow$ \ref{point3} is proved in \citep[Theorem 2.1]{Mohar1991}: the multiplicity of the Laplacian eigenvalue $\mu_n = 0$ is equal to the number of components of graph $G$. 
\end{proof}

From the three conditions above connectedness of the comparison multigraph will be used in our discussion. If some properties are required of graph $G$, it means that only the appropriate subset of ranking problems $(N,A,M)$ is considered.

A graph $G$ is called \emph{bipartite} if its node set $N$ can be divided into two disjoint subsets $U$ and $V$ such that every edge connects a vertex in $U$ to one in $V$. Equivalently, a bipartite graph is a graph without odd-length circles.
Notice that a similar lemma can be stated for the other special structure of the matches matrix: it is block anti-diagonal if and only if the comparison multigraph is bipartite. The equivalence is due to the fact that the objects can be divided into two groups without comparisons inside the groups.

Intuitively, uniqueness of the least squares solution should be provided when all objects $X_i$ and $X_j$ can be compared directly or indirectly, that is, there exists a chain $X_i = X_{k_0}, X_{k_1}, \dots , X_{k_t} = X_j$ such that for each $\ell \in \{ 0,1, \dots , t-1 \}, X_{k_{\ell}}$ has been compared with $X_{k_{\ell+1}}$.

\begin{proposition} \label{Prop31}
The least squares rating $\mathbf{q}$ is unique if and only if comparison multigraph $G$ is connected.
\end{proposition}

\begin{proof}
In the unweighted case, see \citet[Theorem 4]{BozokiFulopRonyai2010}. The same theorem was proved by \citet[p. 426]{KaiserSerlin1978} in a different way.

The general weighted case is examined in \citet{BozokiCsatoRonyaiTapolcai2014} and \citet{Gonzalez-DiazHendrickxLohmann2013}. 
\citet[p. 220]{ChebotarevShamis1999} mention this fact without further discussion. 
\end{proof}

If there is no relation between two groups of objects $N_1$ and $N_2$ as graph $G$ is not connected, then it seems strange to rank them on the same scale. 

Now another solution is given for the least squares problem based on $n \times n$ matrices and $n$-dimensional vectors. In this sense it differs from the proposals of \citet{Gulliksen1956} and \citet{BozokiFulopRonyai2010}, but it is similar to the approach of \citet{KaiserSerlin1978}. 

Let $I$ be the $n \times n$ identity matrix as before, $J$ be the $n \times n$ matrix of $1$'s and, with a slight abuse of notation, $\mathbf{0}$ be both the $n \times n$ matrix and the $n$-dimensional vector of $0$'s.  
We adjust the Laplacian matrix in order to eliminate its zero eigenvalue.

\begin{lemma} \label{Lemma32}
Let $G$ be a connected comparison multigraph. Then $\mu_{n-1} > 0$, the matrix $L + (1/n) J$ is nonsingular with eigenvalues $\mu_1, \mu_2, \dots ,\mu_{n-1}, 1$. If $L^+$ is the Moore-Penrose generalized inverse of $L$, then $\left[ L + (1/n) J \right]^{-1} = L^+ + (1/n) J$.
\end{lemma}

\begin{proof}
This formula is well-known in the literature, it has been rediscovered several times. The first appearance may be in \citet{SharpeStyan1965}, see also \citet[Theorem 10.1.2]{RaoMitra1971} and \citet[Propositions 15 and 16]{ChebotarevAgaev2002}. Here we give a new proof.

The unweighted case is discussed in \citet[Theorems 4 and 5]{GutmanXiao2004}.

For the general weighted version, $\mu_{n-1} > 0$ was proved in Lemma \ref{Lemma31}. It is always possible to choose the Laplacian eigenvectors to be real, normalized and mutually orthogonal. The eigenvector corresponding to $\mu_n = 0$ is of the form $\mathbf{u}^{(n)} = \left[ 1,1, \dots ,1 \right]^\top, \mathbf{u}^{(n)} \in \mathbb{R}^n$. Since the eigenvectors $\mathbf{u}^{(k)}, k = 1,2, \dots ,n-1$, are orthogonal to $\mathbf{u}^{(n)}$, $\sum_{j=1}^n u_j^{(k)} = 0$ is satisfied for all $k = 1,2, \dots ,n-1$. Then the extension of the result of \citet[Theorem 4]{GutmanXiao2004} to multigraphs is obvious as $\left[ L + (1/n) J \right] \mathbf{u}^{(k)} = L \mathbf{u}^{(k)} + (1/n) J \mathbf{u}^{(k)} = L \mathbf{u}^{(k)} + \mathbf{0} = \mu_k \mathbf{u}^{(k)}$ for all $k = 1,2, \dots ,n-1$ and $\left[ L + (1/n) J \right] \mathbf{u}^{(n)} = L \mathbf{u}^{(n)} + (1/n) J \mathbf{u}^{(n)} = \mathbf{0} + (1/n) n \mathbf{u}^{(n)} = \mathbf{u}^{(n)}$, thus the last eigenvalue of $L + (1/n) J$ is equal to $1$ with the corresponding eigenvector $\mathbf{u}^{(n)}$.

\citet{Kwiesielewicz1996} shows that $L L^+ = L^+ L = I - (1/n) J$ is provided in the weighted case, too. It implies that $\left[ L + (1/n) J \right] \left[ L^+ + (1/n) J \right] = L L^+ + \mathbf{0} + \mathbf{0} + (1/n)^2 J^2 = I - (1/n) J + (1/n)^2 n J = I$, since $J L^+ = \mathbf{0}$ \citep[Theorem 4]{Kwiesielewicz1996}, consequently $\left[ L + (1/n) J \right]^{-1} = L^+ + (1/n) J$. 
\end{proof}

\citet{KaiserSerlin1978} use the matrix $L + J$ in the unweighted case to circumvent the singularity of $L$.

\begin{theorem} \label{Theo31}
Let $G$ be a connected comparison multigraph. The unique solution of the least squares problem is
\[
\mathbf{q} = L^+ \mathbf{s} = \left[ L + (1/n) J \right]^{-1} \mathbf{s}.
\]
\end{theorem}

\begin{proof}
Lemma \ref{Lemma32} provides the following equivalent transformations of the least squares problem:
\[
L \mathbf{q} = \mathbf{s} \Leftrightarrow \left[ L + (1/n) J \right] \mathbf{q} = \mathbf{s} + (1/n) J \mathbf{q} \Leftrightarrow \mathbf{q} = \left[ L^+ + (1/n) J \right] \mathbf{s} + (1/n) \left[ L^+ + (1/n) J \right] J \mathbf{q}.
\]
Since $L^+ J = \mathbf{0}$ \citep[Theorem 4]{Kwiesielewicz1996}, $J^2 = nJ$, and $J \mathbf{s} = (\sum_{i=1}^n s_i) \mathbf{e} = \mathbf{0}$ due to $\mathbf{e}^\top \mathbf{s} = \sum_{i=1}^n s_i = 0$:
\[
L \mathbf{q} = \mathbf{s} \Leftrightarrow \mathbf{q} = L^+ \mathbf{s} + (1/n) J \mathbf{q}.
\]
Normalization $J \mathbf{q} = \mathbf{0}$ can only be done if $J L^+ \mathbf{s} = \mathbf{0}$, which is satisfied because $J L^+ = \mathbf{0}$. 
\end{proof}

This solution concept resolves the problem of the singularity of $L$, while simple calculation is preserved since $L^+$ can be obtained through the identity $L^+ L = L L^+ = I - (1/n)J$. Theorem \ref{Theo31} is mainly a technical result, it means a step towards the iterative calculation of the least squares method.

\section{The iterative calculation of the least squares rating} \label{Sec4}

In this section an iterative process is given for the calculation of the optimal least squares solution. \citet{Gulliksen1956} offers a similar approach but his method is based on the choice of an arbitrary solution $\mathbf{q}$ and adjusting it according to the error term $\mathbf{s} - L \mathbf{q}$. Our proposal starts with the first estimation of ratings by the row sums $\mathbf{s}$ which will be updated by the scores of objects compared with it, and so on. It is similar to the rating method called \emph{recursive Buchholz}, defined on the aggregated paired comparison matrix $R$ \citep{Gonzalez-DiazHendrickxLohmann2013}. However, the latter uses an 'average' setting, the modified scores vector $\mathbf{s}'$ and matches matrix $M'$ after division by the number of comparisons $d_i = \sum_{j=1}^n m_{ij}$ for each $X_i$. Interestingly, despite the different approach, the recursive Buchholz ranking coincides with the one obtained from the least squares solution, it gives a rating vector which is the half of $\mathbf{q}$ \citep[Proposition 3.1]{Gonzalez-DiazHendrickxLohmann2013}.

Recursive Buchholz is a special case of the \emph{recursive performance} defined by \citet{Brozos-Vazquezetal2008} where uniqueness is proved for any matches matrix $M$, which is not block diagonal (comparison multigraph $G$ is connected) and not block anti-diagonal ($G$ is not bipartite). It was shown in Section \ref{Sec3} that the first condition is necessary for the uniqueness of the least squares solution as well, while the second one requires some comments. A block anti-diagonal matches matrix represents a comparison structure similar to a 'team tournament' where the objects (players) have two disjoint subsets (teams) such that players in one team do not play against their teammates. Thus the ratings of the players in one team can be calculated only through the ratings of the players in the other team and this cyclic feature thwarts the convergence of the iteration. An analogous problem will also emerge in our discussion.

A \emph{digraph} is an irreflexive directed graph consisting of a finite set of nodes $N$ and a collection of ordered pairs $P$ of these nodes. An edge from node $X_i$ to node $X_j$ represents a dominance relation of the former over the latter, and is represented by $(i,j) \in P$. 
In our setting it may be discussed as a ranking problem $(N,A,M)$, where $N$ is the set of nodes, the elements of the results matrix $A$ are restricted by $a_{ij} \in \{ -1,0,1 \}$ and the matches matrix $M$ is defined by $m_{ij} = 1$ if and only if $\{ (i,j),(j,i) \} \cap P \neq \emptyset$. Furthermore, $a_{ij} = 1 \Leftrightarrow \left[ (i,j) \in P \text{ and } (j,i) \notin P \right]$, $a_{ij} = -1 \Leftrightarrow \left[ (j,i) \in P \text{ and } (i,j) \notin P \right]$, and $a_{ij} = 0 \Leftrightarrow \left[ \{ (i,j),(j,i) \} \cap P = \emptyset \text{ or } (i,j),(j,i) \in P \right]$, but in the latter case there is a match between objects $X_i$ and $X_j$, therefore $m_{ij} = m_{ji} = 1$. This correspondence is clearly not unique; it can be legitimately argued that edges in both directions mean two matches between the associated objects.

\citet{HeringsvanderLaanTalman2005} define the \emph{positional power} of nodes in digraphs and prove that it can be obtained as the limit point of an iterative process. More details about this method will be given later in order to show its common roots with our iterative solution for the least squares method.

\citet{Chebotarev1994} gives a decomposition of the generalized row sum method by the powers of the parameter  $\varepsilon$. Let $\mu_1$ be the greatest eigenvalue of $L$.

\begin{proposition}
For all $0 < \varepsilon < 1/ \mu_1$ the generalized row sum rating vector is:
\[
\mathbf{x} = \left[ \sum_{k=0}^{\infty} \varepsilon^k (-L)^k \right] (1 + \varepsilon mn) \mathbf{s} = (1 + \varepsilon mn) \mathbf{s} - \varepsilon L (1 + \varepsilon mn) \mathbf{s} + \varepsilon^2 L^2 (1 + \varepsilon mn) \mathbf{s} - \varepsilon^3 L^3 (1 + \varepsilon mn) \mathbf{s} + \dots
\]
\end{proposition}

\begin{proof}
See \citet[Property 12]{Chebotarev1994}. 
\end{proof}
In particular,
\[
x_i = s_i + \varepsilon \left[ (mn-d_i) s_i + \sum_{j \neq i} m_{ij} s_i \right] + o(\varepsilon).
\]

A similar decomposition of the least squares rating is based on Theorem \ref{Theo31} and on the \emph{Neumann series} \citep{Neumann1877} of $\left[ L + (1/n) J \right]^{-1}$.

\begin{lemma} \label{Lemma41}
Let $B \in \mathbb{R}^{n \times n}$. The following statements are equivalent:
\begin{enumerate}
\item
The Neumann series $\sum_{k=0}^{\infty} B^k = I + B + B^2 + B^3 + \dots$ converges;
\item
All eigenvalues $\lambda$ of $B$ are in the interior of the unit circle, that is, $\max \{ |\lambda|: \lambda \mathbf{y} = B \mathbf{y} \} < 1$;
\item
$\lim_{n \to \infty} B^n = \mathbf{0}$.
\end{enumerate}
In which case, $(I-B)^{-1}$ exists, and
\[
\left( I-B \right)^{-1} = \sum_{k=0}^{\infty} B^k = I + B + B^2 + B^3 + \dots \, .
\]
\end{lemma}

\begin{proof}
It is a special case of the theorem for Neumann series in \citet[p. 618]{Meyer2000}. 
\end{proof}

In order to decompose the least squares rating vector $\mathbf{q}$, the Neumann series should be applied for $L + (1/n)J$ used in Theorem \ref{Theo31}. Therefore, some results are necessary about its eigenvalues.
According to the \emph{Ger\v sgorin theorem} \citep{Gersgorin1931}, all eigenvalues of the Laplacian matrix $L$ lie within the closed interval $\left[ 0, 2 \mathfrak{d} \right]$, where $\mathfrak{d} = \max \{ d_i: i = 1,2, \dots ,n \}$ is the \emph{maximal number of comparisons} with the other objects. In the unweighted case $\mathfrak{d} \leq n-1$, and for a round-robin ranking problem $\mathfrak{d} = n-1$.

A \emph{regular} graph is a graph such that every vertex has the same degree. A \emph{semiregular bipartite} (or biregular) graph is a bipartite graph, for which every two vertices on the same side of the given partition have the same degree.

\begin{lemma} \label{Lemma42}
Let $G$ be a graph with a decreasing degree sequence $\mathfrak{d} = d_1 \geq d_2 \geq \dots \geq d_{n}$ ($d_i = \sum_{j=1}^n m_{ij} = - \sum_{j \neq i} \ell_{ij}$) and $L$ be the Laplacian matrix of $G$. Then
\[
\mu_1 \leq  2 \mathfrak{d}.
\]
If $G$ is connected, equality holds if and only if $G$ is a regular bipartite graph.
\end{lemma}

\begin{proof}
In the unweighted case, $\mu_1 \leq \max \{ d(u) + d(v) | (u,v) \in E(G) \}$ and equality holds if and only if $G$ is a semiregular bipartite graph \citep{AndersonMorley1985}. It carries over to multigraphs since the number of matches between two objects is nonnegative \citep[Theorem 2.2]{Mohar1991}.
\end{proof}

Notice that if the comparison multigraph $G$ of the ranking problem $(N,A,M)$ is regular bipartite, then $M$ has a block anti-diagonal structure, but the reverse of the implication does not hold.
 
Let us introduce the $n \times n$ real matrix $C$ with $c_{ij} = -\ell_{ij} = m_{ij}$ for all $i \neq j$ and $c_{ii} = \mathfrak{d} - \ell_{ii} = \mathfrak{d} - d_i = \mathfrak{d} - \sum_{j=1}^n m_{ij}$. $C$ is the same as the matches matrix outside the diagonal, where elements are increased (except for the object(s) with maximal comparisons) in order to provide that the sum of all row (and column) is equal. Then $L = \mathfrak{d} I - C$, therefore
\[
\left[ L + (1/n) J \right]^{-1} = \left[ \mathfrak{d} I - C + (1/n) J \right]^{-1} = \frac{1}{\mathfrak{d}} \left[ I -  \frac{1}{\mathfrak{d}} \left( C - \frac{1}{n} J \right) \right]^{-1}.
\]
In the following, stochastic matrix $(1 / \mathfrak{d}) C$ is denoted by $P$.

\begin{theorem} \label{Theo41}
Let the comparison multigraph be connected, and not regular bipartite. The unique solution $\mathbf{q}$ of the least squares problem is
\[
\mathbf{q} = \frac{1}{\mathfrak{d}} \sum_{k=0}^{\infty} P^k \mathbf{s} = \frac{1}{\mathfrak{d}} \left( \mathbf{s} + P \mathbf{s} + P^2 \mathbf{s} + P^3 \mathbf{s} + \dots \right).
\]
\end{theorem}

\begin{proof}
Let $\lambda$ be an eigenvalue of $\frac{1}{\mathfrak{d}} \left( C - \frac{1}{n} J \right)$, namely, $\lambda \mathbf{y} = \frac{1}{\mathfrak{d}} \left( C - \frac{1}{n} J \right) \mathbf{y}$ for some $\mathbf{y}$. It implies that $\mathfrak{d} (1 - \lambda) \mathbf{y} = \mathfrak{d} \left[ I -  \frac{1}{\mathfrak{d}} \left( C - \frac{1}{n} J \right) \right] \mathbf{y} = \left( L + \frac{1}{n} J \right) \mathbf{y}$. From Lemma \ref{Lemma32}, $\mathfrak{d} (1 - \lambda) \in \{ \mu_1, \mu_2, \dots, \mu_{n-1}, 1 \}$. Since $\mu_{n-1} > 0$ also holds, $\mathfrak{d} (1 - \lambda) > 0$, thus $\lambda < 1$. As a consequence of Lemma \ref{Lemma42}, we have $\mathfrak{d} (1 - \lambda) \leq  2\mathfrak{d}$, therefore, $\lambda \geq -1$. According to the condition of Theorem \ref{Theo41}, $G$ is connected, the equality holds if and only if $G$ is a regular bipartite graph, resulting in the statement: $G$ is not a regular bipartite comparison multigraph if and only if $\lambda > -1$.

Hence all eigenvalues fulfil the requirement $-1 < \lambda < 1$, Lemma \ref{Lemma41} can be applied for the matrix $B = \frac{1}{\mathfrak{d}} \left( C - \frac{1}{n} J \right)$. By applying the Neumann series on Theorem \ref{Theo31}, we obtain
\[
\mathbf{q} = \left[ L + (1/n) J \right]^{-1} \mathbf{s} = \frac{1}{\mathfrak{d}} \left[ I -  \frac{1}{\mathfrak{d}} \left( C - \frac{1}{n} J \right) \right]^{-1} \mathbf{s} = \frac{1}{\mathfrak{d}} \sum_{k=0}^{\infty} B^k \mathbf{s} = \frac{1}{\mathfrak{d}} \sum_{k=0}^{\infty} \left( P - \frac{1}{\mathfrak{d} n} J \right)^k \mathbf{s}.
\]
But $J \mathbf{s} = \mathbf{0}$, which leads to $\left( P - \frac{1}{\mathfrak{d} n} J \right)^k \mathbf{s} = P^k \mathbf{s}$, therefore the assertion holds. 
\end{proof}

For ranking purposes, the multiplier $(1/ \mathfrak{d}) > 0$ in the decomposition of $\mathbf{q}$ is irrelevant.
It follows from Theorem \ref{Theo41} that the least squares solution can be obtained as a limit point of an iterative process.

\begin{proposition}
Let the comparison multigraph be connected, and not regular bipartite. The unique solution of the least squares problem is $\mathbf{q} = \lim_{k \to \infty} \mathbf{q}^{(k)}$, where
\[
\mathbf{q}^{(0)} = (1/ \mathfrak{d}) \mathbf{s},
\]
\[
\mathbf{q}^{(k)} = \mathbf{q}^{(k-1)} + \frac{1}{\mathfrak{d}} P^k \mathbf{s}, \quad k = 1,2, \dots \, .
\]
\end{proposition}

\begin{proof}
It is the immediate consequence of Theorem \ref{Theo41}. 
\end{proof}

The iteration process has an interpretation on graphs. In the following description, the multiplier $1 / \mathfrak{d}$ in the decomposition of $\mathbf{q}$ is disregarded for the sake of simplicity. Let $G'$ be a graph identical to the comparison multigraph except that $\mathfrak{d} - d_i$ loops are assigned for object $X_i$. With this modification, balancedness is achieved with the minimal number of loops, at least one node (with the maximal degree) has no loops.
Graph $G'$ is said to be the \emph{balanced comparison multigraph}. It is the same procedure as balancing the weighted graph $G$ by loops in \citet[p. 1495]{Chebotarev2012}, where $G'$ is called the \emph{balanced-graph} of $G$. Note that $G'$ is connected or regular bipartite if and only if $G$ is connected or regular bipartite, respectively.

Initially, all objects (nodes) are endowed with an own estimation of performance $\mathbf{s}$, corresponding to the row sum vector. In the first step, the performance of objects compared with the given one is taken into account through the edges. $P \mathbf{s}$ means the average scores of the objects that were compared with it (weighted by the number of comparisons, that is, the sum of edges between the two objects). The introduction of $\mathfrak{d} - d_i$ loops on $X_i$ provides that the number of objects reachable on 1-long paths is exactly $\mathfrak{d}$. Now strength of objects compared with the given one is added to the original scores to get $\mathbf{s} + P \mathbf{s}$.

In the $k$th step, the average scores of objects available on all $k$-long paths $P^k \mathbf{s}$ is added to the previous rating vector. If $G$ is a connected, and not regular bipartite graph, then this iteration converges to the least squares ranking due to Theorem \ref{Theo41}.
Example \ref{Examp44} illustrates the decomposition of the least squares rating for the ranking problem analyzed in Example \ref{Examp31}.

\begin{figure}
\centering
\caption{The balanced comparison multigraph of Example \ref{Examp44}}
\label{Fig5}
\vspace{0.5cm}
\begin{tikzpicture}[scale=1,auto=center, transform shape, >=triangle 45, every loop/.style={}]
  \node (n1) at (1,10) {$s_1 = 1$};
  \node (n3) at (1,8)  {$s_3 = 0$};
  \node (n5) at (1,6)  {$s_5 = 1$};
  \node (n2) at (7,10) {$s_2 = 1$};
  \node (n4) at (7,8)  {$s_4 = 0$};
  \node (n6) at (7,6)  {$s_6 = -3$};
  \node (n7) at (4,8)  {$s_7 = 0$};

  \foreach \from/\to in {n1/n3,n3/n5,n5/n6,n5/n7,n2/n4,n4/n6,n7/n6}
    \draw (\from) -- (\to);
    
\path (n1) edge[loop] node {} (n1)
	  (n1) edge[loop above] node {} (n1)
	  (n2) edge[loop] node {} (n2)
	  (n2) edge[loop above] node {} (n2)
	  (n3) edge[loop] node {} (n3)
	  (n4) edge[loop] node {} (n4)
	  (n7) edge[loop] node {} (n7);
\end{tikzpicture}

\end{figure}
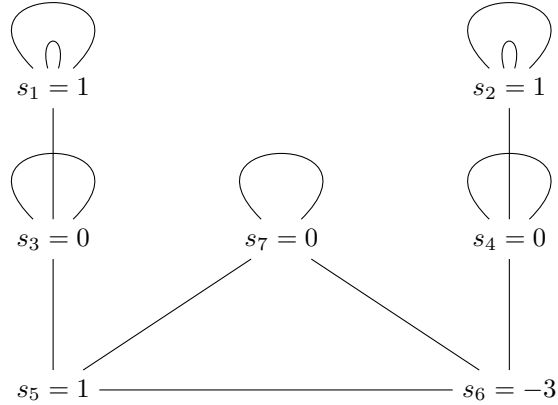
\begin{example} \label{Examp44}
See the preference graph on Figure \ref{Fig1} and its balanced comparison multigraph on Figure \ref{Fig5}. It is an undirected graph, the number of loops are determined by the differences $\mathfrak{d} - d_i$, $\left[ 2,2,1,1,0,0,1 \right]$. Nodes are labelled by the score of the corresponding object. At the start ($\mathfrak{d} \mathbf{q}^{(0)}$), every node gets $s_i$. In the $1$st step ($\mathfrak{d} \mathbf{q}^{(1)}$), the scores of nodes reachable on a 1-long path are added with a multiplier $1 / \mathfrak{d} = 1/3$. For example, in the case of $X_1$ it is $(2 s_1 + s_3)/ \mathfrak{d} = 2/3$.

In the $k$th iteration, the scores of nodes reachable on a $k$-long path are added with a multiplier $(1 / \mathfrak{d})^k$, where the number of scores taken into account is $\mathfrak{d}^k$, analogously. It also implies that the denominator in the fraction of the actual rating $\mathbf{q}^{(k)}_i$ is a divisor of $\mathfrak{d}^{k+1}$ for all $i = 1,2, \dots ,n$. 
Theorem \ref{Theo41} ensures that this process converges if the comparison multigraph $G$ is not a regular bipartite graph. 

The rating vectors obtained in the successive steps of the iteration process are as follows
\[
\setlength{\arraycolsep}{6pt}
\begin{array}{ccrrrrrr}
    \text{Iterated ratings} & \multicolumn{1}{c}{\mathbf{q}^{(0)}}   & \multicolumn{1}{c}{\mathbf{q}^{(1)}}  & \multicolumn{1}{c}{\mathbf{q}^{(2)}}   & \multicolumn{1}{c}{\mathbf{q}^{(3)}}   & \multicolumn{1}{c}{\mathbf{q}^{(10)}}   & \multicolumn{1}{c}{\mathbf{q}^{(50)}}     & \multicolumn{1}{c}{\mathbf{q}} \\
    \midrule
    X_1  & 1/3  & 5/9	& 21/27    & 76/81    & 1.5075  & 1.8057  & 1.8095 \\
    X_2  & 1/3  & 5/9   & 17/27    & 56/81    & 0.6915  & 0.4800  & 0.4762 \\
    X_3  & 0    & 2/9   & 7/27     & 29/81    & 0.6178  & 0.8069  & 0.8095 \\
    X_4  & 0    & -2/9  & -5/27    & -19/81   & -0.3535 & -0.5211 & -0.5238 \\
    X_5  & 1/3	& \multicolumn{1}{c}{0}    & 1/27     & -7/81    & -0.2092 & -0.1912 & -0.1905 \\
    X_6  & -1   & -8/9  & -31/27   & -95/81   & -1.4450 & -1.5231 & -1.5238 \\
    X_7  & 0    & -2/9  & -10/27   & -40/81   & -0.8090 & -0.8571 & -0.8571 \\
\end{array}
\]
It immediately shows the role of comparisons. For example, the scores of $X_1$, $X_2$ and $X_5$ are equal, however, their position in the preference graph is significantly different, which can be seen in the subsequent steps of the iteration. The final ranking emerges only after the 13th step.
\end{example}

Example \ref{Examp44} suggests two observations. The first is that ties are usually eliminated after taking the comparison structure into account, which can be advantageous in practical applications by reducing the demand for tie-breaking rules. The second is the possibly slow convergence: in order to get the final ranking of the objects, long paths also may be necessary to consider causing some difficulties in the interpretation since it is not exactly clear why they still have some importance. Nevertheless, the graph of Example \ref{Examp31} has few edges relative to a round-robin ranking problem, thus it is not surprising that many iteration steps are required.

Theorem \ref{Theo41} has virtually no significance from a computational viewpoint since the least squares problem can be solved with a modest cost of $O(n^3)$ flops \citep{JiangLimYaoYe2011}.

Iterative scoring procedures used for ranking the nodes in a digraph can be traced back to the works of \citet{Wei1952} and \citet{Kendall1955}, called the long-path method by \cite{Laslier1997}. It is based on the right eigenvector corresponding to the largest positive eigenvalue of the adjacency matrix. \citet{MoonPullman1970} shows that the iterative procedure converges to a non-zero vector if the digraph is strongly connected, namely, there exists a path from $X_i$ to $X_j$ if $X_i \neq X_j$. According to \citet{Chebotarev1994} and \citet{HeringsvanderLaanTalman2005}, the severely restricted domain limits the usefulness of this concept.

This drawback is eliminated by the positional power measure \citep{HeringsvanderLaanTalman2005}. Its rating vector $\mathbf{p}$ is the limit point of the sequence
\[
\mathbf{p}^0 = \mathbf{0},
\]
\[
\mathbf{p}^k = T^A \mathbf{e} + \frac{1}{n} T^A \mathbf{p}^{k-1}, \quad k = 1,2, \dots .
\]
Now the first step ($T^A \mathbf{e}$) gives the 'score' of nodes in a digraph, the number of their successors. Subsequently, each node gets a fraction $1/n$ of the previous power of its successors and a fixed amount of 1. \citet{HeringsvanderLaanTalman2005} do not mention the use of the Neumann series explicitly. However, the decomposition in the proof of \citet[Lemma 4.2]{HeringsvanderLaanTalman2005} is based on the equation $\left[ I - (1/n)T^A \right]^{-1} = I + \sum_{k=1}^\infty (1/n)^k \left( T^A \right)^k$ as $\lim_{k \to \infty} (1/n)^k \left( T^A \right)^k = \mathbf{0}$.


Besides these common roots, we have identified three differences between the least squares rating and positional power of nodes in digraphs. The first is in the approach of the two ratings. According to the concepts of \citet{ChebotarevShamis1999}, positional power (as well as the Wei-Kendall method) is a kind of \emph{win-loss combining} procedure distinguishing the wins and the losses of objects, while least squares is a \emph{win-loss unifying} procedure, treating all results uniformly. Here the outcomes of paired comparisons only appear in the results matrix $A$, therefore they influence the ranking through $\mathbf{s}$.

The second difference is the role of iterated ratings: for the positional power they are allocated to the predecessors of nodes, whereas for the least squares, they remain on the objects. On the other side, positional power adds the original score $T^A \mathbf{e}$ to the nodes in each step, while the least squares procedure uses the fixed scores vector $\mathbf{s}$ for the adjustment of previous ratings. 

The third, maybe the most interesting difference is the choice of parameter $1/a$ reflecting the importance of  successors in the digraph. \citet{HeringsvanderLaanTalman2005} define it somewhat arbitrarily as $1/a = 1/n$, the reciprocal of the number of nodes in the digraph, however, the procedure works for any nonnegative numbers less than $1/(n-1)$. It would be interesting to see how this parameter $1/a$ can be increased. The proof of \citet[Theorem 3.1]{HeringsvanderLaanTalman2005} certainly works if all nodes have less than $a$ successors. In such a way the exact definition of $a$ becomes endogenous, similar that of the least squares method, and the procedure will depend on the structure of the digraph.

It was shown in Theorem \ref{Theo41} that for the least squares method the \emph{decay parameter} $1 / \mathfrak{d}$ is determined endogenously by $\mathfrak{d} = \max \{ d_i: i = 1,2, \dots ,n \}$, the maximal number of comparisons of any objects and the iteration process works for all ranking problems $(N,A,M)$ with a connected, and not regular bipartite comparison multigraph. The proposal of \citet{HeringsvanderLaanTalman2005} can also be applied in Theorem \ref{Theo41}, as the convergence clearly holds if the parameter $1/c$ in the iteration is smaller than $1 / \mathfrak{d}$, which is provided if $c > m(n-1)$. For instance, $c = mn$ is a value analogous to the idea of \citet{HeringsvanderLaanTalman2005}, but it obviously differs from the least squares method.

\begin{remark}
In Theorem \ref{Theo41}, the decay parameter $1/ \mathfrak{d}$ is determined endogenously by the decomposition of $L = \mathfrak{d} I - C$. If it becomes larger, the convergence of the Neumann series is not ensured, there will be more critical cases than regular bipartite graphs. If it is smaller, the iteration converges, however, loops will always appear in the balanced comparison multigraph $G'$ and the interpretation becomes more complicated.
\end{remark}

The decomposition of the least squares rating works perfectly for regular graphs without loops. They are characteristic for some applications, like Swiss-system tournaments, and it is unlikely that such a set of comparisons results in a bipartite comparison multigraph.

Other graph interpretations of the least squares method may be possible on the basis of the system of linear equations $\mathbf{q} = L^{+} \mathbf{s}$. For example, a topological interpretation was given for $L^{+}$ in \citet[Theorem 3]{ChebotarevShamis1998b}.

Finally, it is worth to compare the graph interpretation above with the one given for the generalized row sum method by \cite{Shamis1994}. The latter calculates the number of $k$-long routes with an even and odd number of drains (i.e. sequence of edges with some possible loops) between the objects, where $\varepsilon$ represents the importance attributed to indirect connections, that is, the $k$-long routes have a weight of $\varepsilon^k$. It works for all $\varepsilon < 1 / \left[ 2 m (n-1) \right]$.
We think that from a graph-theoretic viewpoint the above interpretation is more simple, however, the appearance of loops remains a weakness.

\section{Concluding remarks} \label{Sec5}

We have shown that the least squares ranking method has a graph interpretation with the exception of some special cases, when the comparison multigraph is a regular bipartite graph. The rating vector can be obtained as the limit of an iteration process based on scores and a decay parameter $1 / \mathfrak{d}$, where $\mathfrak{d} = \max \{ d_i: i = 1,2, \dots ,n \}$ is the maximal number of comparisons, determined endogenously by the matches matrix $M$. 

Aggregation of the results $A = \sum_{p=1}^{m} A^{(p)}$ eliminates a lot of information regarding the outcomes of paired comparisons, for instance, $a_{ij} = a_{ij}^{(1)} + a_{ij}^{(2)}$ can be equal to $0$ by adding both $1$ and $-1$ or $0$ and $0$. We do not know of any ranking methods which, besides the aggregated expected value $a_{ij}$, account for the variance of $a_{ij}^{(p)}$, i.e. the bias from the fact that usually $a_{ij}^{(p)}$ is not equal to the average $a_{ij} / m_{ij}$. However, the difference $a_{ij}^{(p)} - a_{ij} / m_{ij}$ can carry some information about the comparison of $X_i$ and $X_j$: intuitively, their relative ranking seems to be more stable if $a_{ij}^{(p)} \approx a_{ij} / m_{ij}$ for all $p$ where $r_{ij}^{(p)}$ is known. A possible way of addressing this reliability of the paired comparisons is through an adjustment of the number of comparisons between $X_i$ and $X_j$.

Since digraphs can be incorporated in our setting, it is interesting to connect the decomposition of the least squares method with the positional power of nodes in weighted digraphs \citep{HeringsvanderLaanTalman2005}. A \emph{weighted digraph} is defined by the set of nodes and a nonnegative matrix $W \in \mathbb{R}_+^{n \times n}$, where $w_{ij} > 0$ denotes the weight of edge from $X_i$ to $X_j$.
It is natural to choose $a_{ij} = w_{ij} / (w_{ij} + w_{ji})$ and $m_{ij} = w_{ij} + w_{ji}$ but in weighted digraphs edges from a node to itself ($w_{ii} > 0$) are also allowed. It remains an open question what is the relation of the two concepts.

Finally, it follows from Theorem \ref{Theo41} that convergence is ensured for all multipliers less than $1/ \mathfrak{d}$, offering a natural way for the generalization of the least squares method. If $G$ is not a bipartite graph, the parameter can also be increased. Another promising direction may be the change of exponential decay.

\section*{Acknowledgements}
We are grateful to Pavel Yurievich Chebotarev and Julio Gonz\'alez-Díaz for their help in many ways. We also thank S\'andor Boz\'oki, Ferenc Forg\'o, J\'anos F\"ul\"op, J\'ozsef Temesi and two anonymous referees for their comments on earlier versions of the manuscript.

The research was supported by OTKA grant K-77420.

This research was supported by the European Union and the State of Hungary, co-financed by the European Social Fund in the framework of TÁMOP 4.2.4. A/1-11-1-2012-0001 'National Excellence Program'.


\end{document}